\documentclass[12pt]{amsart}
\usepackage[utf8]{inputenc}
\usepackage{graphicx}
\usepackage{color}
\usepackage{array}
\usepackage[text={6.6in,9.1in},centering]{geometry}
\usepackage{amsmath,amsthm}
\usepackage{amssymb,latexsym}
\usepackage{mathrsfs}
\usepackage{cancel}
\usepackage[initials,short-journals,short-months]{amsrefs}
\usepackage[pdfstartview={FitH}]{hyperref}
\usepackage{url}
\usepackage{tikz}
\usepackage{enumitem}
\usepackage{soul}
\usepackage{bbm}
\usepackage[FIGTOPCAP]{subfigure}

\newtheorem{theorem}{Theorem}

\newtheorem{proposition}[theorem]{Proposition}
\newtheorem{corollary}[theorem]{Corollary}

\theoremstyle{definition}

\theoremstyle{remark}
\newtheorem*{remark*}{Remark}
\newtheorem{remark}[theorem]{Remark}

\def\XXint#1#2#3{{
\setbox0=\hbox{$#1{#2#3}{\int}$}
\vcenter{\hbox{$#2#3$}}\kern-.5\wd0}}

\def \CSmooth(#1,#2){\mathcal{C}_{#1,#2}}

\def \Mgale(#1,#2){M_{#1}^{#2}}
\def \Ngale(#1,#2){\mathcal{N}_{#1}^{#2}}

\renewcommand{\theequation}{\thesection.\arabic{equation}}
\numberwithin{equation}{section}

\date{\today}

\renewcommand{\theequation}{\thesection.\arabic{equation}}
\numberwithin{equation}{section}

\title[Self-similar Laplacian with singularly continuous spectrum]{Singularly continuous spectrum of a self-similar Laplacian on the half-line}
\author{Joe P. Chen}
\address{Department of Mathematics, University of Connecticut, Storrs, CT 06269, USA}
\email{joe.p.chen@uconn.edu}
\urladdr{\url{http://homepages.uconn.edu/jpchen}}

\author{Alexander Teplyaev}
\email{teplyaev@math.uconn.edu}
\urladdr{\url{http://homepages.uconn.edu/teplyaev/}}

\thanks{Research partially supported by NSF DMS-1262929.}
\date{\today}
\keywords{Laplacian, spectral decimation, singularly continuous spectrum, Julia set, fractals.}
\subjclass[2010]{47A10, 37F50, 81Q12, 81Q35}

\begin{document}
\renewcommand{\theequation}{\thesection.\arabic{equation}}
\numberwithin{equation}{section}

\begin{abstract}
We investigate the spectrum of the self-similar Laplacian, which generates the so-called ``$pq$ random walk'' on the integer half-line $\mathbb{Z}_+$. Using the method of spectral decimation, we prove that the spectral type of the Laplacian is singularly continuous whenever $p\neq \frac{1}{2}$. This serves as a toy model for generating singularly continuous spectrum, which can be generalized to more complicated settings.
We hope it will provide more insight into Fibonacci-type and other weakly self-similar models.
\end{abstract}

\maketitle

\section{Introduction}

Several models of one-dimensional discrete Schr\"odinger operators have been proved to exhibit purely singular continuous spectrum; see for instance \cites{Simon95,JL99,AvilaDamanik,Quint}.
In this brief paper, we consider a particular family of self-similar Laplacians $\Delta_p$ on the integer half-line $\mathbb{Z}_+$, parametrized by $p\in (0,1)$. 
The parameter $p$ plays the role of the transition probability of a symmetrizable random walk. From the physical point of view, changing $p$ corresponds to changing the contrast ratio of the fractal media. From the mathematical point of view, these Laplacians arise from the study of the unit interval endowed with a fractal measure, and was first addressed by the second author in \cite{TeplyaevSpectralZeta} in the context of spectral zeta functions; see also the related work \cite{DGV12}. 
In this context, the parameter $p$ determines the resistance and measure scaling of the fractal space. 
In particular, we obtain a simple one-parameter family of models for which the \emph{spectral dimension} $d_s$ of $\Delta_p$ (see Remark~\ref{re-ds}) 
varies continuously:  
\begin{equation}\label{e-ds}
	d_s=\dfrac{\log9}{\log\big(1{+}\tfrac2{p(1-p)}\big)}\in (0,1].
\end{equation}
It will be explained that when $p=\frac12$, we recover the classical one-dimensional Laplacian with $d_s=1$. We can also take the direct product of any number of these fractal intervals to construct a fractal of a higher dimension. For instance, a fractal with topological dimension 4 and spectral dimension 2 can be obtained by taking the direct product of 4 one-dimensional intervals, each equipped with a fractal Laplacian $\Delta_p$ with $p(1-p)=\frac1{40}$, or equivalently, $d_s=\frac12$.

The key question addressed in this paper concerns the spectral type of the fractal Laplacian $\Delta_p$. Related questions about wave propagation on this fractal, \emph{viz.} the modes of convergence of discrete wave solutions to the fractal wave solution, were studied in \cites{CNT14, ABCMT}.
For recent physics results, theoretical and experimental, see \cites{A,a1,a2,Dunne,LBFA,adt} and references therein. 
In general, weakly self-similar fractal systems are related to quasicrystals. Although we do not discuss this relation, 
the reader can find   explanations  in \cites{b1,BIST,dl1,dl2,LBFA} and references therein. 
Our long-term motivation comes from the fact that many problems  of fractal nature appear in quantum gravity 
(\emph{e.g.\@} \cites{b2,e,l,r,Lbook}), but the related mathematical physics  
(\emph{e.g.\@} \cite{MST,S1,st,FKS,FST,CMT,HKT,r1,r2,i1,i2,i3,OST}) is not sufficiently developed to approach these problems, mainly because it is hard to tackle problems of fractal geometry and spectral analysis simultaneously. Hence analyzing a straightforward fractal model, such as the one described in this paper, may be of special interest.

Our result is relatively simple because of the minimality of our model, but it relies upon spectral decimation (or spectral similarity) \cites{eigenpaper,MT03} and its connection with the Julia set of a rational function. Parallel ideas have also appeared in the proofs of singularly continuous spectrum for Fibonacci Hamiltonians on $\mathbb{Z}$ (see \cites{d1,d2,d2-,d3,d4,d5,Mei} and references therein), as well as the relation between Julia sets and Jacobi matrices (see 
\cites{BGH1,BGH2,BGM1,BGM2}). 
One of most recent articles on this topics, 
emphasizing the relation between self-similarity and 
singularly continuous spectrum, is \cite{GLN1,GLN}.
We hope that the methods outlined in this paper can be generalized to more complicated settings.

\subsection*{Acknowledgements}
A substantial part of this work was completed and presented at the workshop ``Spectral Properties of Quasicrystals via Analysis, Dynamics, and Geometric Measure Theory''  at the Casa Matem\'{a}tica Oaxaca (CMO).
We thank the Banff International Research Station for Mathematical Innovation and Discovery, and the organizers and participants  for their support.
We are especially grateful to D. Damanik and A. Gorodetski
for many insightful remarks and suggestions.

\section{Main Results}\label{sec:m}

The $pq$-model on $\mathbb{Z}_+$ is defined as follows. Let $p\in (0,1)$ and $q=1-p$. For each $x\in \mathbb{Z}_+ \setminus \{0\}$, let $m(x)$ be the largest natural number $m$ such that $3^m$ divides $x$. Then for all functions $f$ on $\mathbb{Z}_+$, we set
\begin{align}\label{Lp}
	(\Delta_p f)(x) = \left\{\begin{array}{ll} f(0)-f(1), & \text{if}~x=0 \\
		f(x)-qf(x-1)-pf(x+1), &\text{if}~3^{-m(x)}x \equiv 1~\pmod 3 \\
		f(x) - pf(x-1)-qf(x+1), &\text{if}~3^{-m(x)}x \equiv 2~\pmod 3
	\end{array}\right..
\end{align}
Observe that the Laplacian $\Delta_p$ generates a nearest neighbor ``$pq$ random walk'' on $\mathbb{Z}_+$ as shown in Figure \ref{figRW}.

\begin{figure}[htp]
	\begin{center}
	\includegraphics{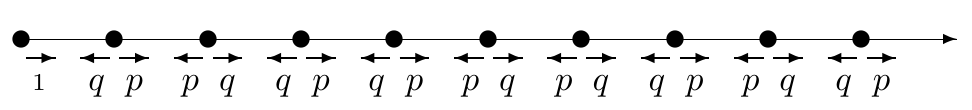}
	\end{center}
	\caption{Transition probabilities in the $pq$ random walk. Here $p\in (0,1)$ and $q=1-p$.}
	\label{figRW}
\end{figure} If $p=\frac{1}{2}$, we recover the symmetric simple random walk on the half-line with reflection at the origin, and $\Delta_{\frac{1}{2}}$, the classical one-dimensional Laplacian, is self-adjoint on $\ell^2(\mathbb{Z}_+)$. If $p\neq \frac{1}{2}$, $\Delta_p$ is not self-adjoint on $\ell^2(\mathbb{Z}_+)$. That said, we can identify the symmetrizing measure for $\Delta_p$ using the theory of Markov chains (see \emph{e.g.\@} \cite{DurrettEOSP}*{Ch.\@ 1}). One can readily verify that $\Delta_p$ generates an irreducible Markov chain on $\mathbb{Z}_+$ whose transition probabilities satisfy $p(x,y)=0$ whenever $|x-y| >1$, \emph{i.e.,} it is a birth-and-death chain. 
As a result, one can explicitly compute the invariant measure $\pi$ by iteratively solving the equation $\pi(y) = \sum_{x\in \mathbb{Z}_+} \pi(x) p(x,y)$. 
This means that the reversibility condition $\pi(x) p(x,y)=\pi(y)p(y,x)$ holds for every $x,y \in \mathbb{Z}_+$, which implies that $\Delta_p$ is symmetric with respect to $\pi$. 
In our example, $\pi$ essentially coincides with a multiple of the
discretization  of the
fractal measure described in \cites{TeplyaevSpectralZeta, ABCMT} ($\pi$ is a $\sigma$-finite but not finite  measure on $\mathbb{Z}_+$). Moreover, 
 we have  the relation $$\pi(x)=\pi(3x),$$ 
 which is easy to verify by induction for all $x\in\mathbb{Z}_+$. This last property allows us to transfer our definitions and 
 main results from $\ell^2(\mathbb{Z}_+)$ to $L^2(\mathbb{Z}_+,\pi)$ in what follows.

Our main result is
\begin{theorem}
\label{thm:scspec}
If $p\neq \frac{1}{2}$, the Laplacian $\Delta_p$, regarded as an operator on either $\ell^2(\mathbb{Z}_+)$ or $L^2(\mathbb{Z}_+,\pi)$, has purely singularly continuous spectrum.
The spectrum is the Julia set of the polynomial $R(z)$ in (\ref{eq:specdecDeltap}), which is a topological Cantor set of Lebesgue measure zero.
\end{theorem}

It is well-known that the spectrum of $\Delta_{\frac12}$, the classical one-dimensional Laplacian on $\mathbb{Z}_+$, is the interval $[0,2]$, and is absolutely continuous. So in a sense, there is a ``phase transition'' in the spectral type of $\Delta_p$ as $p$ varies through $\frac12$, going from a singular spectrum to an absolutely continuous spectrum and back to a singular spectrum.

\begin{remark}
We note that, following now standard techniques for the so-called Sturmian potentials (see \cites{BIST,d0,dl1,dl2}),
one can hope to extend this result
to two-sided models on $\mathbb Z$.
However,
there is a technical difficulty in the fact that
 that the
density of the symmetrizing measure
$ \pi $
on $ \mathbb{Z}_+ $
with respect to the counting measure
is not bounded from above and below.
\end{remark}

\begin{remark}
	The disconnectedness of the Julia set $\mathcal{J}(R)$ implies that the Laplacian spectrum has infinitely many (large) gaps, which is a salient feature of many symmetric finitely ramified fractals \cite{HSTZ12}. As a result, the summability of Fourier series is better on these fractals than that on Euclidean space \cite{StrichartzFourierGaps}.
\end{remark}

\begin{remark}\label{re-ds}
The classical notion of spectral dimension $d_s$ is introduced in  
\cites{AO} for discrete Laplacians on infinite graphs, and in \cites{KL,Lbook} for the corresponding continuous Laplacians 
on compact Dirichlet metric spaces. We note that not all authors agree with this notion of $d_s$; see \cite{StrichartzFSFs} for a detailed discussion. 

In the context of our paper, the spectral dimension $d_s$ is understood as follows. Take the sequence of Laplacians $\Delta_p$ restricted 
to the segment $[0,3^m] \cap \mathbb{Z}_+$. One can estimate the lowest non-zero eigenvalue of $\Delta_p$
by the inverse composition powers $R^{-\circ m}(2)$, which behave, up to a constant, as $\big(R'(0)\big)^{-m}$. 
Here $R(z)$ is the spectral decimation function given in \eqref{eq:specdecDeltap}, and $R'(0)=1+\frac2{pq}$. The spectral dimension is then given by $d_s = \frac{2\log M}{\log R'(0)}$, where $M$ stands for the rate of volume growth between successive fractal approximations. In our case $M=3$, so we recover \eqref{e-ds}.
This method of calculating the spectral dimension of a self-similar  Laplacian
which admits spectral decimation is discussed in 
\cites{d1,FukushimaShima,t,TeplyaevSpectralZeta}. 

Alternatively, according to the approach of Kigami and Lapidus \cites{KL}, 
under certain assumptions that are satisfied in our case, 
the spectral dimension of a self-similar set with resistance scaling factors 
$r_j$ and measure scaling factors $m_j$ is defined as the unique number $d_s$ that satisfies 
$$\sum_{j=1}^N(r_jm_j)^{d_s/2}=1.$$
For our fractal Laplacian $\Delta_p$, the resistance scaling factors are
$
r_1=r_3=\frac{p}{1+p}$ and $
r_2=\frac{q}{1+p}, 
$
and 
measure weights are 
$ 
m_1=m_3=\frac{q}{1+q}$ and $ 
m_2=\frac{p}{1+q} 
$ (for more details see \cite{TeplyaevSpectralZeta}). 
From these it is direct to verify that $d_s$ agrees with \eqref{e-ds}. 
A more recent work \cite{ABCMT} also discusses the probabilistic meaning of this spectral dimension in terms of heat kernel estimates, but it is not needed for the present paper. 
\end{remark}

\section{Proof of Theorem \ref{thm:scspec}}

Throughout the section, $\rho(A)$ and $\sigma(A)$ stands for the resolvent set and the spectrum of an operator $A$, respectively.

\subsection{Spectral decimation}

We briefly review the necessary ingredients from spectral decimaion that will be used in the proof. Spectral decimation originated from \cites{b1,RammalToulouse}, and was implemented on the Sierpinski gasket in \cites{Shima91, FukushimaShima,t} and on post-critically finite fractals in \cites{Shima96}. Here we follow \cite{MT03}*{Definition 2.1} (see also \cite{eigenpaper} for more information). Let $\mathcal{H}$ and $\mathcal{H}_0$ be Hilbert spaces, and $H$ (resp. $H_0$) be operators on $\mathcal{H}$ (resp. $\mathcal{H}_0$). We say that $H$ is \emph{spectrally similar} to $H_0$ with functions $\varphi_0,\varphi_1: \rho(H) \to\mathbb{C}$ if there exists a (partial) isometry $U: \mathcal{H}_0\to\mathcal{H}$ such that
\begin{align}
U^*(H-z)^{-1} U = \left(\varphi_0(z) H_0 - \varphi_1(z)\right)^{-1}
\end{align}
whenever both sides are defined.

For concreteness, we will specialize spectral similarity to the case where $\mathcal{H}_0$ is a closed subspace of $\mathcal{H}$, and $U^*=:P_0$ is the orthogonal projection from $\mathcal{H}$ to $\mathcal{H}_0$. Let $\mathcal{H}_1$ be the orthogonal complement of $\mathcal{H}_0$ in $\mathcal{H}$, and $P_1 = I-P_0$ be the orthogonal projection from $\mathcal{H}$ to $\mathcal{H}_1$. Define $I_0: \mathcal{H}_0\to\mathcal{H}_0$, $X: \mathcal{H}_0\to\mathcal{H}_1$, $\bar{X}: \mathcal{H}_1\to\mathcal{H}_0$, and $Q: \mathcal{H}_1\to\mathcal{H}_1$ by $I_0 =P_0 H P_0^*$, $X=P_1 H P_0^*$, $\bar{X}=P_0 H P_1^*$, and $Q=P_1 H P_1^*$. In other words, $H$ has the following block structure with respect to the representation $\mathcal{H}= \mathcal{H}_0\oplus \mathcal{H}_1$:
\begin{align}
H = \begin{pmatrix} I_0 & \bar{X} \\ X & Q\end{pmatrix}.
\end{align}

  According to \cite{MT03}*{Corollary 3.4}, without loss of generality, we may assume that $\varphi_0$ and $\varphi_1$ are defined on $\rho(Q)$. Then by \cite{MT03}*{Definition 3.5}, we introduce the \emph{exceptional set} of the spectrally similar operators $H$ and $H_0$ as follows:
\begin{align}
\mathfrak{E}(H,H_0) = \{z\in \mathbb{C}: z\in\sigma(Q)~\text{or}~\varphi_0(z)=0\}.
\end{align}
Let $R(z) = \varphi_1(z)/\varphi_0(z)$ whenever $\varphi_0(z)\neq 0$.

The key result we need is
\begin{proposition}[\cite{MT03}*{Theorem 3.6}]
\label{prop:specdec}
Let $H$ be spectrally similar to $H_0$ on $\mathcal{H}_0$, and $z\notin \mathfrak{E}(H,H_0)$. Then
\begin{enumerate}
\item $R(z)\in \rho(H_0)$ if and only if $z\in \rho(H)$.
\item $R(z)$ is an eigenvalue of $H_0$ if and only if $z$ is an eigenvalue of $H$. Moreover, there is a one-to-one map $f_0 \mapsto f= f_0-(Q-z)^{-1}Xf_0$ from the eigenspace of $H_0$ corresponding to $R(z)$ onto the eigenspace of $H$ corresponding to $z$.
\end{enumerate}
\end{proposition}

\subsection{Spectral decimation for $\Delta_p$}

We now apply the above framework to the operator $\Delta_p$ on $\ell^2(\mathbb{Z}_+)$. Here we put $\mathcal{H} = \ell^2(\mathbb{Z}_+)$ and $\mathcal{H}_0 = \ell^2(3\mathbb{Z}_+)$. Then $$U^*: \ell^2(\mathbb{Z}_+) \to \ell^2(3\mathbb{Z}_+)$$ is the orthogonal projection defined by 
$$(U^* f)(3x)=f(3x).
$$
Moreover, following the idea of Bellissard
\cite[page 125]{b1}, we can define a dilation operator
	$D:\ell^2(3\mathbb{Z}_+) \to \ell^2( \mathbb{Z}_+)$,
$$D f(x)=f(3x),$$
and its co-isometric adjoint  
$D^*:\ell^2(\mathbb{Z}_+) \to \ell^2( 3\mathbb{Z}_+)$,
$$(D^* f)(3x)=f(x).
$$
Then we define the operator $\Delta_p^+$ on $\ell^2(3\mathbb{Z}_+)$ to be $$\Delta_p^+ = D^* \Delta_p D.$$
By definition, $\Delta_p^+$ on $\ell^2(3\mathbb{Z}_+)$
is 
isometrically equivalent to $\Delta_p$ on $\ell^2(\mathbb{Z}_+)$. 
Moreover, 
$\Delta_p^+$ on  $L^2(3\mathbb{Z}_+,\pi)$
is 
isometrically equivalent to $\Delta_p$ on  $L^2(\mathbb{Z}_+,\pi)$
because of the relation
$\pi(x)=\pi(3x)$.

\begin{proposition}[Spectral decimation for $\Delta_p$]
\label{prop:specdecDeltap}
The operator $\Delta_p$ on $\ell^2(\mathbb{Z}_+)$ is spectrally similar to $\Delta_p^+$ on $\ell^2(3\mathbb{Z}_+)$ with functions
\begin{align}
\label{eq:specdecDeltap}
\varphi_0(z)=\frac{pq}{(1-z)^2-p^2}, \quad \varphi_1(z) = R(z) \varphi_0(z), \quad R(z) = \frac{z(z^2-3z+(2+pq))}{pq}.
\end{align}
\end{proposition}

Proposition \ref{prop:specdec} was essentially proved in \cites{TeplyaevSpectralZeta, ABCMT}. It follows from taking the Schur complement of $\Delta_p$ with respect to the block corresponding to projection of functions onto $\mathbb{Z}_+\setminus (3\mathbb{Z}_+)$. For the reader's convenience, we give a self-contained proof in  Appendix \ref{appA}.

Next, we identify the exceptional set of $\Delta_p$ and $\Delta_p^+$. Note that $\varphi_0(z) \neq 0$ for all $z\in \mathbb{C}$. As for the operator $Q:\mathcal{H}_1 \to \mathcal{H}_1$, (\ref{Lp}) yields
\begin{align*}
(Qf)(3x+1) &= f(3x+1) - p f(3x+2) \\
(Qf)(3x+2) &= f(3x+2)-pf(3x+1) 
\end{align*}
for each $x\in \mathbb{Z}_+$. This means that $Q$, as a matrix with respect to the natural basis of delta functions on $\mathbb{Z}_+\setminus 3\mathbb{Z}_+$, is a block diagonal matrix consisting of $2\times 2$ blocks
\begin{align*}
\begin{pmatrix} Q(3x+1, 3x+1) & Q(3x+1, 3x+2) \\ Q(3x+2,3x+1) & Q(3x+2,3x+2)\end{pmatrix} = \begin{pmatrix} 1 & -p \\ -p & 1\end{pmatrix},~ x\in \mathbb{Z}_+.
\end{align*}
From this it is easy to deduce that $\sigma(Q) = \{1+p, 1-p\}$. Thus $\mathfrak{E}(\Delta_p, \Delta_p^+) = \{1+p, 1-p\}$. 

The next result is a direct consequence of Proposition \ref{prop:specdec}.
\begin{proposition}
\label{prop:iterate}
Suppose $z\notin \{1+p, 1-p\}$. Then
\begin{enumerate}
\item $R(z) \in \rho(\Delta_p^+) = \rho(\Delta_p)$ 
if and only if 
$z\in \rho(\Delta_p)$.
\item $R(z)$ is an eigenvalue of $\Delta_p^+$ if and only if $z$ is an eigenvalue of $\Delta_p$. Furthermore, there is an injection
from the eigenspace of $\Delta_p^+$ with eigenvalue $R(z)$ to the eigenspace of $\Delta_p$ with eigenvalue $z$, given by $u^+\mapsto u$, $u(x) =u^+(3x)$.
\end{enumerate}
\end{proposition}

Actually we can say more. Due to the self-similarity of the Laplacian $\Delta_p$, $\Delta_p^+$ has the same spectrum as $\Delta_p$, and in fact they are isomorphic as bounded symmetrizable operators. This observation combined with Proposition \ref{prop:iterate} leads to

\begin{corollary}[Spectral self-similarity of $\Delta_p$]
\label{cor:sss}
Suppose $z\notin\{1+p, 1-p\}$. Then
\begin{enumerate}
\item $R(z) \in \rho(\Delta_p)$ if and only if $z\in \rho(\Delta_p)$.
\item $R(z)$ is an eigenvalue of $\Delta_p$ if and only if $z$ is an eigenvalue of $\Delta_p$.
\end{enumerate}
\end{corollary}


It remains to resolve the status of the exceptional points.

\begin{proposition}
\label{prop:exceptional}
$1\pm p\in \sigma(\Delta_p)$.
\end{proposition} 
\begin{proof}
The operator $\Delta_p-z$ has the block structure
\begin{align}
\Delta_p-z = \begin{pmatrix}
I_0-z & \bar{X} \\ X & Q-z \end{pmatrix}
\end{align}
with respect to the representation $\mathcal{H} = \mathcal{H}_0 \oplus \mathcal{H}_1$.
It is direct to verify that $\Delta_p-z$ is invertible if and only if both $Q-z$ and the Schur complement $(I_0-z)-\bar{X}(Q-z)^{-1}X$ are invertible. Since $1\pm p \in \sigma(Q)$ by the computation prior to Proposition \ref{prop:iterate}, we conclude that $1\pm p\in \sigma(\Delta_p)$.
\end{proof}

\begin{remark*}Figure~\ref{RGraph} shows the graph of $R$. From the point of view of dynamics on the Riemann sphere $\hat{\mathbb{C}} := \mathbb{C}\cup \{\infty\}$, the polynomial $R$ has four fixed points, $0$, $1$, $2$, and $\infty$. The first three are repulsive, since $
	|R'(0)|=|R'(2)| = \frac{2+pq}{pq} > 1$ and $|R'(1)| = \left|1-\frac{1}{pq}\right|\geq 3$, while $\infty$ is superattracting.
The spectral decimation function $R$ in (\ref{eq:specdecDeltap}) depend on $pq=p(1-p)$ and is symmetric in $p$ and $q$. So according to Corollary \ref{cor:sss}, the spectrum of $\Delta_{1-p}$, as a compact subset of $\mathbb R$, is equal to the spectrum of $\Delta_{p}$. However, the eigenfunctions of $\Delta_{1-p}$ do not coincide with the eigenfunctions of $\Delta_p$; see \cite{TeplyaevSpectralZeta} for details. 
If we assume for a moment that 
$p\in (0, \frac{1}{2}]$, then 
the preimage of $[0,2]$ under $R$ is
\begin{align*}
[0,~p] \cup [q,~1+p] \cup [1+q,~2].
\end{align*}
If 
$p\in ( \frac{1}{2},1]$, then 
the preimage of $[0,2]$ under $R$ is
\begin{align*}
[0,~q] \cup [p,~1+q] \cup [1+p,~2].
\end{align*}
\begin{figure}[ht]
\begin{center}
\includegraphics{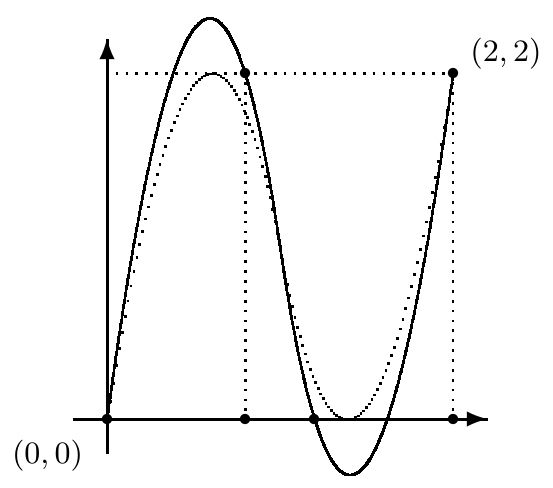}
\end{center}
\caption{The graph of the cubic polynomial $R(z)$
associated with the Laplacian $\Delta_p$.}
\label{RGraph}
\end{figure}In particular, when $p=\frac{1}{2}$, $R(z)$ is the cubic Chebyshev polynomial, and the preimage of $[0,2]$ under $R$ is the entire interval $[0,2]$. The graph of   $R(z)$ in the case $p=\frac12$ is illustrated by the curved dotted line in 
Figure~\ref{RGraph}, and the solid curved line sketches  
the graph of $R(z)$ when $p\neq\frac12$. 
\end{remark*}

We now recall some facts from complex dynamics (see \emph{e.g.\@} \cite{Milnor}*{\S4}). The Fatou set $\mathcal{F}(g)$ of a nonconstant holomorphic function $g$ on the Riemann sphere $\hat{\mathbb{C}}$ is the domain in which the family of iterates $\{g^{\circ n}\}_n$ converges uniformly on compact subsets. The complement of the Fatou set in $\hat{\mathbb{C}}$ is the Julia set $\mathcal{J}(g)$. 
Both $\mathcal{F}(g)$ and $\mathcal{J}(g)$ are fully invariant under $g$: that is, $g^{-1}(\mathcal{F}(g)) = \mathcal{F}(g)$ and $g^{-1}(\mathcal{J}(g)) = \mathcal{J}(g)$. Moreover, $\mathcal{J}(g)$ is a closed subset of $\hat{\mathbb{C}}$.

For the spectral decimation function $R$ in (\ref{eq:specdecDeltap}), we have the following characterization of the Julia set $\mathcal{J}(R)$, which is standard in complex dynamics (see \cite{Milnor}):

\begin{proposition}
\label{prop:Julia}
The Julia set $\mathcal{J}(R)$ of the cubic polynomial map $R$ in (\ref{eq:specdecDeltap}) is contained in $[0,2]$. If $p=\frac{1}{2}$, $\mathcal{J}(R) = [0,2]$. If $p \neq \frac{1}{2}$, $\mathcal{J}(R)$ is a Cantor set of Lebesgue measure zero.
\end{proposition}

By \cite{Milnor}*{Lemma 4.6}, $\{0,1,2\}\subset \mathcal{J}(R)$ because they are repulsive fixed points of $R$. 

\begin{theorem}
\label{thm:specJulia}
$\sigma(\Delta_p)=\mathcal{J}(R)$.
\end{theorem}

\begin{proof}
We readily verify that $0\in \sigma(\Delta_p)$ (its corresponding formal eigenfunction is $f_0 \equiv 1$) and $2\in \sigma(\Delta_p)$ (eigenfunction is $f_2(x) = (-1)^x,~x\in\mathbb{Z}_+$). Also, $R(1-p)=2$ and $R(1+p)=0$. This combined with Proposition \ref{prop:exceptional} allows us to strengthen Corollary \ref{cor:sss} to the following statement:
\begin{align}
\label{newcor8}
\text{For every $z\in \mathbb{C}$, $R(z) \in \rho(\Delta_p)$ if and only if $z \in \rho(\Delta_p)$.}
\end{align}

Now we show $\mathcal{J}(R) \subset \sigma(\Delta_p)$.
By \eqref{newcor8}, all pre-iterates of $0$ under $R$ lie in $\sigma(\Delta_p)$. Since $\sigma(\Delta_p)$ is closed,
\begin{align}
\label{sigmaJ1}
\overline{\bigcup_{n=0}^\infty R^{\circ -n}(0)} \subset \sigma(\Delta_p).
\end{align}
Meanwhile, $0\in \mathcal{J}(R)$, and by \cite{Milnor}*{Corollary 4.13}, the set of all pre-iterates of a point in the Julia set is everywhere dense in the Julia set. This implies that
\begin{align}
\label{sigmaJ2}
\overline{\bigcup_{n=0}^\infty R^{\circ -n}(0)} = \mathcal{J}(R).
\end{align}
It follows from \eqref{sigmaJ1} and \eqref{sigmaJ2} that $\mathcal{J}(R) \subset \sigma(\Delta_p)$.

Next we show $\sigma(\Delta_p) \subset \mathcal{J}(R)$. Let $z\in \sigma(\Delta_p)$. By \eqref{newcor8}, $R^{\circ n}(z) \in \sigma(\Delta_p)$ for all $n\in \mathbb{Z}_+$. Since $\sigma(\Delta_p)$ is compact (and hence bounded), and the only attracting fixed point of $R$ is $\infty$, it follows that $z$ cannot be in the Fatou set (which contains the basin of attraction of $\infty$, and is thus unbounded in $\mathbb{C}$). So $z\in \mathcal{J}(R)$.
\end{proof}

\begin{remark*}
It is instructive to compare the proof above with the proof of \cite{MT03}*{Theorem 5.8}, which relates the spectrum of the Laplacian on a symmetric self-similar graph to the Julia set of the corresponding spectral decimation function. We summarize the main differences between the two proofs.

In the proof above, we first took advantage of the condition $R(\mathfrak{E})\subset \sigma(\Delta_p)$, and deduced the full invariance of $\sigma(\Delta_p)$ under $R$. To prove $\mathcal{J}(R)\subset \sigma(\Delta_p)$, we identified a point in $\mathcal{J}(R)\cap \sigma(\Delta_p)$, and used the full invariance. To prove $\sigma(\Delta_p) \subset \mathcal{J}(R)$, we used the full invariance, and the fact that $\infty$ is the only attracting fixed point of $R$.

In the proof of \cite{MT03}*{Theorem 5.8}, the setting was more general, and in particular, it does not always hold that the spectrum $\sigma(\Delta)$ is fully invariant under the corresponding spectral decimation function $R_\Delta$. To prove $\mathcal{J}(R_\Delta) \subset \sigma(\Delta)$, the authors used the fact that $0 \in \mathcal{J}(R_\Delta) \cap \sigma(\Delta)$, as well as the fact that $0$ is not an isolated eigenvalue of $\Delta$. The proof of the other inclusion, $\sigma(\Delta) \subset \mathcal{J}(R_\Delta) \cup \mathcal{D}_\infty$ (where $\mathcal{D}_\infty$ is defined therein), follows from a standard argument in complex dynamics. 
\end{remark*}

\subsection{The main proof} 

We now have all the ingredients to prove Theorem \ref{thm:scspec}.

\begin{proof}[Proof of Theorem \ref{thm:scspec} for $\Delta_p$ on $\ell^2(\mathbb{Z}_+)$]
First of all, Proposition \ref{prop:Julia} and Theorem \ref{thm:specJulia} together imply that when $p\neq\frac{1}{2}$, $\Delta_p$ has no absolutely continuous spectrum. So we turn to the point spectrum of $\Delta_p$. Theorem \ref{thm:specJulia} says that it suffices to consider points in $\mathcal{J}(R)$.

 Let us assume that a formal eigenfunction $f_z$ of $\Delta_p$ with eigenvalue $z$ exists, that is, $\Delta_p f_z=z f_z$. We see that
\begin{align}\label{f_z}
f_z(1) = (1-z) f_z(0),
\end{align}
 by solving the eigenvalue equation at the origin; and if $f_z(0)=0$, then $f_z \equiv 0$, by solving the eigenvalue equation iteratively along $\mathbb{Z}_+$. So it is enough to consider the case $f_z(0) \neq 0$. Upon dividing $f_z$ by $f_z(0)$, we may set $f_z(0)=1$ without loss of generality.  
Let us first establish that none of the fixed points $\{0,1,2\}$ of $R$ is an eigenvalue of $\Delta_p$. By iterating the eigenvalue equation along $\mathbb{Z}_+$, it is easy to verify that
\begin{align*}
f_0\equiv 1; ~~
|f_1(4x)| = \left(\frac{q}{p}\right)^2 ~\text{for all}~x\in \mathbb{Z}_+ \setminus \{0\};~~\text{and}~
f_2(x) = (-1)^x~\text{for all}~x\in \mathbb{Z}_+.
\end{align*}
Therefore $f_z \notin \ell^2(\mathbb{Z}_+)$ for $z\in \{0,1,2\}$, so none of the fixed points is an eigenvalue of $\Delta_p$. By Corollary \ref{cor:sss},  any preimage of any fixed point under $R$ cannot be an eigenvalue, either.

Next, if we take $z\in \mathcal{J}(R)$ which is not a preimage of a fixed point of $R$, then by the definition and the basic properties of the Julia set, the sequence of iterates $\{R^{\circ n}(z)\}_n$ does not have a limit. From the eigenfunction statement in Proposition \ref{prop:iterate} and (\ref{f_z}), $f_z(3^n)=1-R^{\circ n}(z)$. Hence $\sum_{n=0}^\infty [f_z(3^n)]^2$ is divergent, which means that $f_z \notin \ell^2(\mathbb{Z}_+)$. This proves that $\Delta_p$ has no point spectrum.

We conclude that $\sigma(\Delta_p)$ has purely singularly continuous spectrum. The rest of Theorem \ref{thm:scspec} follows from Proposition \ref{prop:Julia} and Theorem \ref{thm:specJulia}.
\end{proof}

\begin{proof}[Proof of Theorem \ref{thm:scspec} for $\Delta_p$ on $L^2(\mathbb{Z}_+,\pi)$]
All the preceding arguments still hold, except that we need to check that none of the formal eigenfunctions is in $L^2(\mathbb{Z}_+,\pi)$.
By the self-similarity of the invariant measure $\pi$, it is direct to verify that $\pi(3^n)$ are identical for all $n\in \mathbb{Z}_+$. Upon replacing $\sum_{n=0}^\infty |f_z(3^n)]^2$ in the previous proof by $\sum_{n=0}^\infty [f_z(3^n)]^2 \pi(3^n)$, we see that the lack of square summability of eigenfunctions in $\ell^2(\mathbb{Z}_+)$ also holds true in $L^2(\mathbb{Z}_+, \pi)$.
\end{proof}

\begin{remark*}
As a consequence of the proof, neither of the exceptional points $1\pm p$ is an eigenvalue of $\Delta_p$. This distinguishes the $pq$-model on $\mathbb{Z}_+$ from most of the other models which admit spectral decimation (see \cites{eigenpaper,eigen2,Shima96,Ketal}), such as the infinite Sierpinski gasket \cite{t}, where there are exceptional points which are eigenvalues of the corresponding Laplacian.
\end{remark*}

\appendix

\section{Proof of Proposition \ref{prop:specdec}} \label{appA}

Let us divide $\mathbb{Z}_+$ into two disjoint subspaces $3\mathbb{Z}_+$ and $\mathbb{Z}_+\setminus 3\mathbb{Z}_+$. Then for $z\in \mathbb{C}$, the operator $\Delta_p-z$ acting on functions on $\mathbb{Z}_+$ can be represented in block matrix form
\begin{align}
\Delta_p-z = \begin{pmatrix}
I_0-z & \bar{X} \\ X & Q-z \end{pmatrix},
\end{align}
where
\begin{align*}
 I_0-z&: \{\text{functions on}~3\mathbb{Z}_+\} \to \{\text{functions on}~3\mathbb{Z}_+\},\\
\bar{X} &: \{\text{functions on}~\mathbb{Z}_+\setminus 3\mathbb{Z}_+\} \to \{\text{functions on}~3\mathbb{Z}_+\},\\
X &: \{\text{functions on}~ 3\mathbb{Z}_+\} \to \{\text{functions on}~\mathbb{Z}_+\setminus 3\mathbb{Z}_+\},\\
Q-z &: \{\text{functions on}~\mathbb{Z}_+\setminus 3\mathbb{Z}_+\} \to \{\text{functions on}~\mathbb{Z}_+\setminus 3\mathbb{Z}_+\}.\\
\end{align*}
The Schur complement $S(z)$ of $\Delta_p-z$ with respect to the block corresponding to functions on $\mathbb{Z}_+\setminus 3\mathbb{Z}_+$ is then given by
\begin{align}
S(z)=
( 
I_0-z) - \bar{X}(Q-z)^{-1} X,
\end{align}
We claim that this equals $\varphi_0(z)(\Delta_p^+ - R(z))$ as an operator acting on functions on $3\mathbb{Z}_+$.
More formally, we consider the matrices of operators
with respect to the natural basis of delta functions on
$\mathbb{Z}_+$.

To compute $S(z)$, let us observe that $I_0-z$ is a diagonal matrix with all diagonal elements equal to $1-z$; $\bar{X}$ has nonzero matrix elements $\bar{X}(0,1)=-1$, $\bar{X}(3x,3x-1) = -q$ (resp. $-p$) and $\bar{X}(3x,3x+1)=-p$ (resp. $-q$) if $3^{-m(3x)} (3x)\equiv 1 \pmod 3$ (resp. if $3^{-m(3x)} (3x)\equiv 2 \pmod 3$); $X$ has nonzero matrix elements $X(3x,3x\pm 1)=-q$ for all $x\in \mathbb{Z}_+$; and $Q-z$ is a block diagonal matrix consisting of $2\times 2$ blocks
\begin{align*}
\begin{pmatrix} (Q-z)(3x+1, 3x+1) & (Q-z)(3x+1, 3x+2) \\ (Q-z)(3x+2,3x+1) & (Q-z)(3x+2,3x+2)\end{pmatrix} = \begin{pmatrix} 1-z & -p \\ -p & 1-z\end{pmatrix},~ x\in \mathbb{Z}_+.
\end{align*}
Since $Q-z$ is block diagonal, it is easy to see that it has an inverse $(Q-z)^{-1}$ whenever $z\notin \{1-p, 1+p\}$. $(Q-z)^{-1}$ is a block diagonal matrix consisting of $2\times 2$ blocks
\begin{align*}
\begin{pmatrix} (Q-z)^{-1}(3x+1, 3x+1) & (Q-z)^{-1}(3x+1, 3x+2) \\ (Q-z)^{-1}(3x+2,3x+1) & (Q-z)^{-1}(3x+2,3x+2)\end{pmatrix} = \frac{1}{(1-z)^2-p^2}\begin{pmatrix} 1-z & p \\ p & 1-z\end{pmatrix} 
.
\end{align*}
After some algebra, we verify that $\bar{X}(Q-z)^{-1}X$ has all diagonal elements equal to $\frac{q(1-z)}{(1-z)^2-p^2}$, and off-diagonal elements
\begin{align*}
(\bar{X}(Q-z)^{-1}X)(3x,3y) = -\frac{pq}{(1-z)^2-p^2} \Delta_p(x,y),~x,y\in\mathbb{Z}_+, ~x\neq y.
\end{align*}
Therefore $(I_0-z)-\bar{X}(Q-z)^{-1} X$ has all diagonal elements equal to $\varphi_0(z)[1-R(z)]$, and off-diagonal elements $\varphi_0(z) \Delta_p^+(x,y)$ in the $(3x,3y)$-entry. This proves the claim. Since $\Delta_p -z$ is invertible if and only if both $Q-z$ and the Schur complement $(I_0-z)-\bar{X}(Q-z)^{-1}X$ are invertible, the claim implies Proposition \ref{prop:specdec}.

\end{document}